\documentclass[runningheads]{llncs}
\usepackage[T1]{fontenc}
\usepackage{amsfonts}
\usepackage{amsmath}
\usepackage{graphicx}
\usepackage{algorithm}
\usepackage{ wasysym }
\usepackage{algpseudocode}

\begin{document}
\title{Quantum Circuit for Quantum Fourier Transform for Arbitrary Qubit Connectivity Graphs}
%
%
\author{Kamil Khadiev\inst{1}\orcidID{0000-0002-5151-9908} \and
Aliya Khadieva\inst{1,2}\orcidID{0000-0003-4125-2151} \and
Vadim Sagitov\inst{1}
\and
Kamil Khasanov\inst{1}
}
\authorrunning{K. Khadiev et al.}
%
\institute{Kazan Federal University, Kazan, Tatarstan, Russia \and
Univerisy of Latvia, Riga, Latvia\\
\email{kamilhadi@gmail.com}}
\maketitle              
\begin{abstract}
In the paper, we consider quantum circuits for the Quantum Fourier Transform (QFT) algorithm. The QFT algorithm is a very popular technique used in many quantum algorithms.
We present a generic method for constructing quantum circuits for this algorithm implementing on quantum devices with restrictions. Many quantum devices (for example, based on superconductors) have restrictions on applying two-qubit gates. These restrictions are presented by a qubit connectivity graph.  Typically, researchers consider only the linear nearest neighbor (LNN) architecture of the qubit connection, but current devices have more complex graphs. We present a method for arbitrary connected graphs that minimizes the number of CNOT gates in the circuit for implementing on such architecture.

We compare quantum circuits built by our algorithm with existing quantum circuits optimized for specific graphs that are Linear-nearest-neighbor (LNN) architecture, ``sun'' (a cycle with tails, presented by the 16-qubit IBMQ device) and ``two joint suns'' (two joint cycles with tails, presented by the 27-qubit IBMQ device). Our generic method gives similar results with existing optimized circuits for ``sun'' and ``two joint suns'' architectures, and a circuit with slightly more CNOT gates for the LNN architecture. At the same time, our method allows us to construct a circuit for arbitrary connected graphs. 
\keywords{QFT, Fourier transform, quantum circuit, NP-hard problem}
\end{abstract}
%
%
\section{Introduction}
\emph{Quantum computing} \cite{nc2010,a2017,aazksw2019part1} is one of the hot topics in computer science of the last decades.
There are many problems in which quantum algorithms outperform the best known classical ones \cite{quantumzoo}. 
One of the well-known computational techniques used in many quantum algorithms is the Quantum Fourier Transform (QFT) \cite{k1995}.
It is used in quantum addition \cite{d2000}, quantum phase estimation (QPE) \cite{k1995}, quantum amplitude estimation (QAE)\cite{bhmt2002}, the algorithm for solving linear systems of equations \cite{hhl2009}, Shor’s factoring algorithm \cite{s1999}, and others. 

In this paper, we are interested in the circuit-based implementation of this algorithm on quantum devices. We are focusing on minimization of two-qubit quantum gates in such a circuit because they are the most ``expensive'' gates to implement. Many types of quantum computers (for example, quantum devices based on superconductors) do not allow us to apply two-qubit gates to an arbitrary pair of qubits. They have a specific architecture of qubits connectivity that are represented by a qubit connectivity graph. Vertices of the graph correspond to qubits, and two-qubit gates can be applied only to qubits corresponding to vertices connected by an edge. In this paper, we focus on the number of CNOT gates in a quantum circuit for the QFT algorithm for devices with a specific qubit connectivity graph. Namely, CNOT is a two-qubit gate that is a quantum analogue of ``excluding or'' operation for classical computation. Let the CNOT cost of a circuit be the number of CNOT gates in the circuit. The CNOT cost of a circuit implementation in a linear nearest-neighbor (LNN) architecture (where the graph is just a chain) was explored by Park and Ahn in \cite{park2023reducing}. They presented a circuit for the QFT algorithm that has $n^2+n-4$ CNOT cost, where $n$ is the number of qubits. It improved the previous results of \cite{nc2010,fdh2004,swd2011,kds2017,bbwdr2019,best1996,tko2007,pa2022}. At the same time, as the authors mentioned, their technique cannot be generalized to more complex graphs. In \cite{k2024aliya}, Khadieva suggested a quantum circuit for a more complex architecture that is a cycle with tails (like a ``sun'' or ``two joint suns''). The CNOT cost of this circuit is $1.5n^2$. In \cite{kkcw2025}, Khadiev et al. suggested a generic method for an arbitrary connected graph.

Here we present a general method that allows us to develop a quantum circuit of the QFT algorithm for an arbitrary connected graph for qubit connectivity. Our algorithm gives a better result compared to \cite{kkcw2025} with respect to the CNOT cost. We define an NP-hard problem called the (3,2,1)-covering path problem that is a modification of the Shortest covering path problem \cite{cpr1994}, the Hamiltonian path problem, and the Travelling salesman problem. We construct our circuit based on the solution of the problem.
The solution uses a dynamic programming approach. The time complexity of the algorithm for constructing the circuit is $O((m+n)2^n)$, where $n$ is the number of qubits and $m$ is the number of edges in the qubit connectivity graph. Additionally, we suggest an approximate solution of the (3,2,1)-covering path problem that has $O((m+n)\log n)$ time complexity.

The constructed circuit has the CNOT cost in the range between $n^2-2n-2$ and $2n^2-2n-2$ 
depending on the complexity of the graph. The result is better than the circuit from \cite{kkcw2025} whose maximum possible CNOT cost is $3n^2-3n$. In addition, we compare our results with circuits for specific graphs. In the case of LNN, the CNOT cost is $1.5n^2-2.5n-1$ that is $1.5$ times larger than the result of \cite{park2023reducing} and the same as the circuit of \cite{k2024aliya}. For more complex graphs such as 16-qubit Falcon r4P and 27-qubit Falcon r5.11 architectures of IBMQ, which is a cycle with tails (like a ``sun'') or its modifications, our generic technique gives the same CNOT cost as the CNOT cost of the circuit \cite{k2024aliya} that was specially constructed for these architectures. In all these cases, our result gives a better circuit than \cite{kkcw2025}. The difference is about 5\%.

The structure of this paper is the following.
Section \ref{sec:prelims} describes the required notations and preliminaries. Graph theory tools are presented in Section \ref{sec:tools}.
 The circuit for the Quantum Fourier Transform algorithm is discussed in Section \ref{sec:qft}. The final Section \ref{sec:concl} concludes the paper and contains some open questions.

\section{Preliminaries}\label{sec:prelims}
\subsection{Graph Theory}
Let us consider an undirected unweighted graph $G=(V,E)$, where $V$ is the set of vertices and $E$ is the set of undirected edges. Let $n=|V|$ be the number of vertices, and $m=|E|$ be the number of edges. 

A non-simple path $P$  is a sequence of vertices $(v_{i_1},\dots,v_{i_h})$ that are connected by edges, that is $(v_{i_j},v_{i_{j+1}})\in E$ for all $j\in\{1,\dots,h-1\}$. Note that a non-simple path can contain duplicates.
Let the length of the path be the number of edges in the path, $len(P)=h-1$.

A path $P=(v_{i_1},\dots,v_{i_h})$ is called simple if there are no duplicates among $v_{i_1},\dots,v_{i_h}$. 
The distance $dist(v,u)$ is the length of the shortest path between vertices $v$ and $u$. Typically, when we say just a ``path'', we mean a ``simple path''. 

Let $\textsc{Neighbors}(v)$ be a list of neighbors for a vertex $v$, i.e., $\textsc{Neighbors}(v)=(u_{i_1},\dots,u_{i_k})$ such that $(v,u_{i_j})\in E$, and $|\textsc{Neighbors}(v)|=k$ is the length of the list.


\subsection{Quantum circuits}\label{sec:qcirc}
Quantum circuits consist of qubits and a sequence of gates applied to these qubits. A state of a qubit is a column-vector from ${\cal H}^2$ Hilbert space. It can be represented by $a_0|0\rangle+a_1|1\rangle$, where $a_0,a_1$ are complex numbers such that $|a_0|^2+|a_1|^2=1$, and $|0\rangle$ and $|1\rangle$ are unit vectors. Here we use the Dirac notation. A state of $n$ qubits is represented by a column-vector from ${\cal H}^{2^n}$ Hilbert space. It can be represented by $\sum_{i=0}^{2^n-1}a_i|i\rangle$, where $a_i$ is a complex number such that $\sum_{i=0}^{2^n-1}|a_i|^2=1$, and $|0\rangle,\dots |2^n-1\rangle$ are unit vectors. Graphically, on a circuit, qubits are presented as parallel lines. 

As basic gates, we consider the following ones:

  $H=\frac{1}{\sqrt{2}}\begin{pmatrix}
1 & 1 \\
1 & -1 
\end{pmatrix}$,     $X=\begin{pmatrix}
0 & 1 \\
1 & 0 
\end{pmatrix}$,
 $R_y(\xi)=\begin{pmatrix}
cos(\xi/2) & -sin(\xi/2) \\
sin(\xi/2) & cos(\xi/2) 
\end{pmatrix}$,

$R_z(\xi)=\begin{pmatrix}
e^{\frac{i\xi}{2}} & 0 \\
0 & e^{-\frac{i\xi}{2}} 
\end{pmatrix}$,
$CNOT=\begin{pmatrix}
1 & 0 & 0 & 0\\
0 & 1 & 0 & 0\\
0 & 0 & 0 & 1\\
0 & 0 & 1 & 0 
\end{pmatrix}$.  

Additionally, we consider four non-basic gates

 $R_k=\begin{pmatrix}
1 & 0 \\
0 & e^{\frac{i\pi}{2^{k-1}}} 
\end{pmatrix}$,
$CR_k=\begin{pmatrix}
1 & 0 & 0 & 0\\
0 & 1 & 0 & 0\\
0 & 0 & 1 & 0\\
0 & 0 & 0 & e^{\frac{i\pi}{2^{k-1}}} 
\end{pmatrix}$,

$CR_z(\xi)=\begin{pmatrix}
1 & 0 & 0 & 0\\
0 & 1 & 0 & 0\\
0 & 0 & e^{\frac{i\xi}{2}}  & 0\\
0 & 0 & 0 & e^{-\frac{i\xi}{2}} 
\end{pmatrix}$, 
$SWAP=\begin{pmatrix}
1 & 0 & 0 & 0\\
0 & 0 & 1 & 0\\
0 & 1 & 0 & 0\\
0 & 0 & 0 & 1 
\end{pmatrix}$,
%

The reader can find more information about quantum circuits in \cite{nc2010,aazksw2019part1,k2022lecturenotes}


\section{(3,2,1)-Covering Path Problem as a Tool}\label{sec:tools}
Let us consider an undirected unweighted connected graph $G=(V,E)$ such that $n=|V|$ is a number of vertices and $m=|E|$ is a number of edges.

In this section, we consider the ``(3,2,1)-Covering Path'' problem ((3,2,1)-CPP or (3,2,1)-CP problem) that is a modification of the well-known shortest covering path problem  (SCPP problem)\cite{cpr1994}. The description of (3,2,1)-CPP is presented below.

The ``(3,2,1)-Shortest Covering Path'' problem ((3,2,1)-CPP or (3,2,1)-CP problem) is defined as follows. Let $P=(v_{i_1},\dots,v_{i_k})$ be a \textbf{non-simple}  path. We say that the path covers all visiting vertices and vertices that are connected with visited vertices by one edge. Formally, the path $P$ covers a set of vertices $R(P)$ such that any vertex $v$ from this set is either
\begin{itemize}
    \item $v$ belongs to $P$ (there is $j\in\{1,\dots,k\}$ such that $v=v_{i_j}$);
    \item $v$ is connected with a vertex from $P$ (there is $j\in\{1,\dots,k\}$ such that $(v,v_{i_j})\in E$).
\end{itemize}
Let $B(P)=R(P)\backslash\{v_{i_1},\dots,v_{i_k}\}$, i.e. they are vertices connected with visited vertices by one edge. 

If the path $P$ covers all the vertices ($R(P)=V$), then we call it a 1-covering path or just a covering path. For a 1-covering path, we define a cost function that is $cost(P)=3(len(P)-1)+2|B(P)|$. The solution of the (3,2,1)-CP problem is the 1-covering path that minimizes the cost function. We call the solution (3,2,1)-covering path.

As the SCP problem, the (3,2,1)-CP problem has a strong connection with the Hamiltonian path problem and the Travelling salesman problem \cite{cormen2001}. Any connected graph has a (3,2,1)-covering path.

The decision version of the SCP problem is NP-complete \cite{cpr1994}. The Travelling salesman problem (TSP) is NP-hard. Similarly, by polynomial reduction of TSP to (3,2,1)-CPP, we can show that it is NP-hard.  


Let us estimate the maximum possible length of a covering path.

\begin{lemma}\label{lm:len-wnsh}
The length of a covering path in a connected graph $G$ of $n$ vertices is at most $2n-3$.
\end{lemma}
\begin{proof}
    Let us consider a spanning tree of the graph $G=(V,E)$. It is a tree $T=(V,E')$, where $E'\subset E$. 
We can construct a non-simple path $P$ that is the Euler tour \cite{cormen2001} of the tree $T$ but does not visit the leaves of the tree. The path covers all the vertices of the graph $G$, but it maybe be does not minimize the cost. Each edge (except edges incident to leafs)  in the tour is visited at most twice (in the up and down direction).  Therefore, the length of the path $len(P)\leq 2n-\ell$, where $\ell$ is the number of leaves, and $\ell\geq 2$. So, we obtain the bound for the number of vertices in the path $2n-2$, and for the length of the path, the bound is $2n-3$.
\end{proof}

Let us present the algorithm for the (3,2,1)-CP problem.
Firstly, let us present a procedure $\textsc{ShortestPaths}(G)$ that constructs two $n\times n$-matrices $W$ and $A$ by a graph $G$. 
Elements of the matrix $W$ are lengths of the shortest paths between each pair of vertices in $G$, i.e. $W[v,u]=dist(v,u)$.
The matrix $A$ represents the shortest paths between the vertices of $G$. The element $A[v,u]$ is the last vertex in the shortest path between $v$ and $u$. In other words, if $t=A[v,u]$, then $P_{v,u}=P_{v,t}\circ u$, where $P_{v,u}$ is the shortest path between $v$ and $u$. Based on this fact, we can present a procedure $\textsc{GetShortestPath}(v,u)$ that computes $P_{v,u}$ using the matrix $A$. Note that the implementation does not add the first element of the path $P_{v,u}$ because we do not need it in our algorithm. The implementation of the procedure is presented in Algorithm \ref{alg:getpath}. (See Appendix \ref{apx:getpath})


We can construct these two matrices using $n$ invocations of the Breadth First Search (BFS) algorithm \cite{cormen2001}. The total time complexity for constructing the matrices is $O(n^3)$. The algorithm for constructing $A$ and $W$ is presented in Appendix \ref{apx:floyd} for completeness of presentation.

Let us define a function $D:2^{V}\times V\to\{0,\dots,n,\infty\}$ such that $D(S,v)$ is the length of the shortest path $P$ that visits all the vertices of $S$ and the last vertex is $v$. Formally, $P=(v_{i_1},\dots,v_{i_k})$, $ v_{i_k}=v$, $S\subset\{v_{i_1},\dots,v_{i_k}\}$. If there is no such path, then $D(S,v)=\infty$. Note that the path $P$ is non-simple, and it can visit some vertex from $V\backslash S$.

Let us present an algorithm for computing $D(S,v)$ for each $S\in2^V$ and $v\in S$. It is easy to see that $\title{D}(\{v\},v)=0$ for each $v\in V$. For other pairs $(S,v)$ we compute it using the following statement $D(S,v)=\min\{D(S\backslash\{v\},u)+W[u,v]: u\in S\}$. 

To construct the path itself, we define a function $F:2^{V}\times V\to V\cup\{NULL\}$ such that $F(S,v)$ is the vertex that precedes $v$ in the shortest path that visits all vertices of $S$. Formally,  $F(S,v)=\min\{i:D(S\backslash\{v\},v_i)+W[v_i,v]=D(S,v), (v_i,v)\in E\}$. If there is no such vertex $v_i$, then $F(S,v)=NULL$.
So, we can compute $F(S,v)$ together with $D(S,v)$, $F(S,v)=u$, if $u=argmin\{D(S\backslash\{v\},u)+W[u,v]: u\in S\}$. If $D(S,v)=\infty$, then $F(S,v)=NULL$.

This idea allows us to define a recursive procedure $\textsc{ComputeD}(G,v)$ whose implementation is presented in Algorithm \ref{alg:ComputeD}. 

\begin{algorithm}[H]
\caption{Implementation of $\textsc{ComputeD}(S,v)$}\label{alg:ComputeD}
\begin{algorithmic}
\If{$S=\{v\}$}
        \State $D(S,v)\gets 0$, $F(S,v)\gets NULL$
\Else 
    \State $D(S,v)\gets \infty$, $F(S,v)\gets NULL$
    \For{$u\in S$}
            \If{$D(S\backslash\{v\},u)$ is not computed}
            \State $\textsc{ComputeD}(S\backslash\{v\},u)$
            \EndIf
            \If{$D(S\backslash\{v\},u)+W[u,v]<D(S,v)$}
                \State $D(S,v)\gets D(S\backslash\{v\},u)+W[u,v]$, $F(S,v)\gets u$
            \EndIf        
    \EndFor
\EndIf
\end{algorithmic}
\end{algorithm}

Let us present the procedure $\textsc{GetNSPath}(S,v)$ that returns the path that visits all vertices of $S$ and ends in $v$. The procedure collects the path using $\textsc{GetShortestPath}$ between the vertices obtained from $F$.
The implementation of $\textsc{GetNSPath}(S,v)$ is presented in Algorithm \ref{alg:getpath2}. (See Appendix \ref{apx:getpath2}).

Furthermore, we define a function $C:2^V\to\{0,1\}$ such that $C(S)=1$ iff $V=S\cup\{v:v\in V\backslash S, $ and there is $u\in S$ such that $(u,v)\in E\}$. In other words, $C(S)=1$ if all vertices of $V\backslash S$ are connected to vertices of $S$ by one edge. Let us define a procedure $\textsc{ComputeC}(G)$ that computes the function $C$. For this reason, we compute a set $R=S\cup \bigcup_{v\in S}\{u: u\in \textsc{Neighbors}(v)\}$, and check if $R=V$. The equivalent condition is $|R|=n$. We do it for each set $S\in2^V$. The implementation of the procedure is presented in Algorithm \ref{alg:checkWHP}.

\begin{algorithm}[H]
\caption{Implementation of $\textsc{ComputeC}(G)$}\label{alg:checkWHP}
\begin{algorithmic}
\For{$S\in 2^V$} 
\State $R\gets S$
\For{$v\in S$}
\For{$u\in \textsc{Neighbors}(v)$}
\State $R\gets R\cup\{u\}$
\EndFor
\EndFor
\If{$|R|=n$}
\State $C(S)\gets 1$
\Else
\State $C(S)\gets 0$
\EndIf
\EndFor
\end{algorithmic}
\end{algorithm}


Now we are ready to define the whole algorithm for the (3,2,1)-CP problem. Firstly, we form the functions $D$, and $F$. For each $S$ that satisfies $C(S)=1$, we choose the path $P$ such that
\begin{itemize}
    \item $P=\textsc{GetNSPath}(S,v)$ is the shortest path that visits all the vertices of $S$ for some $v\in S$;
    \item the value $3 len(P)+2|V\backslash S|=3 D(S,v)-2|V\backslash S|=3 D(S,v)-2(n-|S|)$ is minimal.
\end{itemize} Note that $P$ can visit not only the vertices of $S$. That is why we choose the largest $S$ for the shortest path $P$. It visits only vertices from $S$ in that case, the value $3 len(P)+2|V\backslash S|$ is the cost of the corresponding path, and the minimization of this value is the target.

Let $\textsc{ThreeTwoOneCP}(G)$ be the procedure that returns the target path for the  (3,2,1)-CP problem. The implementation of the procedure is presented in Algorithm \ref{alg:wnshpath}. 
The correctness and complexity of the algorithm is discussed in Theorem \ref{th:wnshpath}

\begin{algorithm}[H]
\caption{Implementation of $\textsc{ThreeTwoOneCP}(G)$}\label{alg:wnshpath}
\begin{algorithmic}
\State $\textsc{ShortestPaths}(G)$
\State $\textsc{ComputeC}(G)$
\State $S'\gets \emptyset, v'\gets NULL, cost\gets \infty$ 
\For{$S\in 2^V$}
\For{$v\in S$}
\State $\textsc{ComputeD}(S,v)$

\If{$C(S)=1$}
\If{$cost>3D(S,v)+2(n-|S|)$ or ($cost>3D(S,v)+2(n-|S|)$ and $|S|>|S'|$) }
\State $cost\gets 3D(S,v)+2(n-|S|), S'\gets S, v'\gets v$
\EndIf
\EndIf
\EndFor
\EndFor
\State $P\gets \textsc{GetNSPath}(S',v')$
\State \Return $P$
\end{algorithmic}
\end{algorithm}

\begin{theorem}\label{th:wnshpath}
    The presented algorithm solves the (3,2,1)-CP problem, and the time complexity is $O((m+n)2^n)$. 
\end{theorem}
\begin{proof}
    Let us show the correctness of the algorithm. Suppose that the algorithm finds the shortest path $P$ that visits all vertices of $S$ such that $C(S)=1$, $S$ is the largest for this length of $P$, and the cost is minimal. Assume that there is a $1$-covering path $P'=(v_{i_1},\dots,v_{i_{k'}})$ that has a lower cost than $P$. Let $S'=\{v_{i_1},\dots,v_{i_{k'}}\}$, then $\textsc{GetNSPath}(S',v_{i_{k'}})=P'$. It means $cost(P')=3len(P')+2|V\backslash S'|=3D(S',v_{i_{k'}})+2(n-S')>3D(S,v_{i_{k}})+2(n-S)=cost(P)$ because $P$ has the smallest value $3D(S,v_{i_{k}})+2(n-S)=cost(P)$ among all paths computed by $\textsc{GetNSPath}(S,v)$. This claim contradicts the assumption $cost(P)>cost(P')$.

    The procedure $\textsc{ComputeD}$ is invoked once for each subset $S\in 2^V$ and vertex $v\in V$. The time complexity of all invocations of the procedure is $O((m+n)\cdot 2^n)$.  
    The time complexity of the \textsc{ShortestPaths} procedure is $O(n^3)$.
    The time complexity for the procedure $\textsc{ComputeC}$ is $O((m+n)2^n)$ because we check all subsets $S\in 2^V$ and check at most $m$ edges of the graph for each subset.
    
    The complexity of $\textsc{GetNSPath}$ is $O(n)$ because the maximal length of the path is $2n$ due to Lemma \ref{lm:len-wnsh}.

    So, the total complexity is $O(n^3+(m+n)\cdot 2^n+(m+n)\cdot 2^n+n)=O((m+n)2^n)$.
\end{proof}

\subsection{Approximate Algorithm for (3,2,1)-Covering Path Problem}

We are planning to use the solution of the problem for optimization of a circuit for the QFT algorithm. So for big $n$, the current solution is too slow.

Due to the strong connection of the (3,2,1)-CP problem with the Travelling salesman problem (TSP) and the Shortst covering path problem (SCPP), we can use heuristic algorithms, for example, Ant colony optimization \cite{dg1997}, or greedy algorithms like \cite{jm1997} that are used for TSP or algorithms used for SCPP \cite{cpr1994}. 

Here we present a fast approximate solution to the problem that can be used for practical applications.

Let us define two subtasks.
\begin{itemize}
    \item The Connected Dominating Set problem (CDS problem). For a given graph $G=(V,E)$, we want to find a connected set $S$ of minimal size such that $V=S\cup B$, where $B=\{u: u\in \textsc{Neighbors}(v)$ for some $v\in S\}$. Informally, each vertex of the graph either belongs to $S$ or is connected to a vertex from $S$ by one edge.
    \item  For a given weighed graph $G'=(V',G')$, the shortest non-simple path that visits all vertices of the graph at least once.
\end{itemize}

The first problem can be solved using a $(ln \Delta + 3)$-approximating algorithm from \cite{gk1998}, where $ \Delta = max\{|\textsc{Neighbors}(v)|: v\in V\}$ is the maximal number of neighbors of a vertex from $V$. Here, $\alpha$-approximating algorithm means that the result is at most $\alpha$ times bigger than the solution. The properties of the algorithm are described in the following lemma.
\begin{lemma}[\cite{gk1998}]\label{lm:cds}
    There is an $(ln \Delta + 3)$-approximate algorithm for the CDS problem. The time complexity of the algorithm is $O((n+m)\log n)$
\end{lemma}

The second problem can be solved by the Christofides–Serdyukov algorithm analogue \cite{c2022,s1978,vs2020}. Let us consider a spanning tree of the graph $G=(V,E)$. It is a tree $T=(V,E')$, where $E'\subset E$. 
We can construct a non-simple path $P$ that is the Euler tour \cite{cormen2001} of the tree $T$. The path visits all the vertices of the graph $G$, but possibly it is not the shortest. The length of the path is $2|V|-2$. The length of the minimal possible path that visits all vertices is at least $|V|-1$.  So, the algorithm gives us at most $2$ times longer path. The solution is a $2$-approximating solution to the second problem.
\begin{lemma}\label{lm:tcsns}
    The time complexity of the presented $2$-approximate algorithm for searching the shortest non-simple path that visits all vertices of the graph at least once is $O(|V|+|E|)$
\end{lemma}
\begin{proof}
    The spanning tree can be constructed using the depth-first search algorithm with $O(|V|+|E|)$ time complexity \cite{cormen2001}. The Euler tour \cite{cormen2001} can also be done with $O(|V|+|E|)$ time complexity.
\end{proof}

So, the whole algorithm is two steps: 
\begin{itemize}
    \item Step 1. Constructing the smallest connected domain $S$ of the graph $G$. Then consider the subgraph $G(S)=(S,E(S))$, where $E(S)\subset E$ are the edges of $G$ that connect only the vertices from $S$. We use the $(ln \Delta + 3)$-approximate algorithm from Lemma \ref{lm:cds}.
    \item Step 2. We construct a path that visits all vertices at least once in the graph $G(S)$.  We use the $2$-approximate algorithm from Lemma \ref{lm:tcsns}.
\end{itemize}

We claim that the presented algorithm solves the (3,2,1)-CP problem and it is a $2(ln \Delta + 3)$-approximate algorithm.

\begin{theorem}
    The presented algorithm solves the (3,2,1)-CP problem, it is a $2(ln \Delta + 3)$-approximate algorithm, and the time complexity is $O((n+m)\log n)$.
\end{theorem}
\begin{proof}
Let us consider the solution $P=(v_{i_1},\dots,v_{i_k})$ for the (3,2,1)-CP problem for some graph $G=(V,E)$. The set $S=\{v_{i_1},\dots,v_{i_k}\}$ is the set of vertices visited by $P$. Note that all vertices of the graph are either belongs to $V$ or connected to a vertex from $S$ with one edge. Let $S_d$ be the solution of the CDS problem for the graph. Therefore, the size $|S|\geq |S_d|$.

The cost of the path $cost(P)=3len(P)+2|V\backslash S|\geq 3|S|+|V\backslash S|=|S|+2|V|=|S|+2n\geq |S_d|+2n$.

Let us consider the solution obtained by the approximate solution to the problem.

Let $S'_d$ be the approximate solution of the first part (to the CDS problem). So, $|S'_d|\leq (ln\Delta +3)|S_d|$.

Let the path $P'$ be the approximate solution of the second part (the shortest non-simple path that visits all vertices of $S'_d$ at least once). The length of the path is $len(P')\leq 2|S'_d|\leq 2(ln\Delta +3)|S_d|$.

The cost of the path $cost(P')=3len(P')+2|V\backslash S'_d|\leq 2(ln\Delta +3)|S_d| +2n-2|S'_d|\leq 2(ln\Delta +3)|S_d| +2n-2|S_d|=2(ln\Delta +2)|S_d| +2n\leq 2(ln\Delta +2)|S_d| +2(ln\Delta +2)\cdot 2n=2(ln\Delta +2)(|S_d|+2n)$.

So, we can say, that $cost(P')\leq 2(ln\Delta +2)(|S_d|+2n)$, and $cost(P)\geq (|S_d|+2n)$. Therefore, $cost(P')\leq cost(P)\cdot 2(ln\Delta +2)$.

The time complexity of the solution is $O((n+m)\log n)$ for the first part, and $O(|S_d'|+E(S_d'))=O(n+m)$ for the second part. The total time complexity is $O((n+m)\log n)$.
\end{proof}


\section{Method for Constructing a Circuit for Quantum Fourier Transform}\label{sec:qft}

Let us consider a quantum device with some qubit connectivity graph $G=(V,E)$. We assume that $G$ is a connected graph. Here we present a method that allows us to construct a circuit that implements the Quantum Fourier Transform (QFT) algorithm on this device. More information on the QFT algorithm can be found in Appendix \ref{apx:qft}. If we do not have restrictions for applying two-qubit gates (when $G$ is a complete graph, for instance), then the circuit is presented in Figure \ref{fig:cqft}.
\begin{figure}[H]
\includegraphics[width=0.5\textwidth]{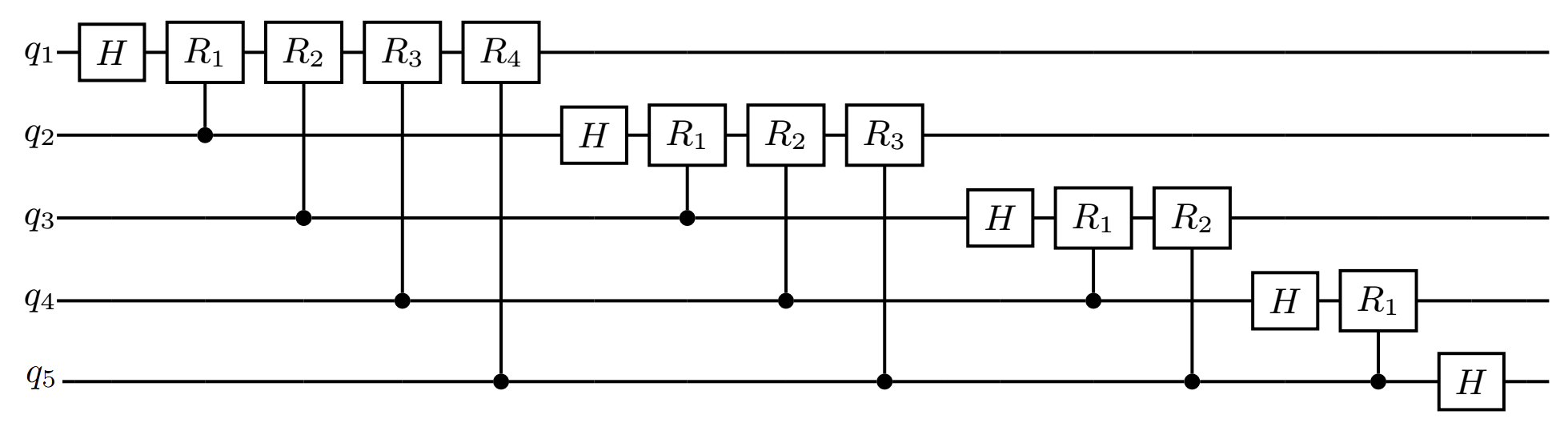}
\caption{\label{fig:cqft} A quantum circuit for Quantum Fourier Transform algorithm for fully connected $5$ qubits}
\end{figure}

We can split the circuit for the QFT algorithm into a series of control phase gates cascades depending on the target qubit for control phase operations. The $r$-th cascade uses $q_r$ as the target qubit (Figure \ref{fig:qft-cascade}).

\begin{figure}[H]
\includegraphics[width=0.5\textwidth]{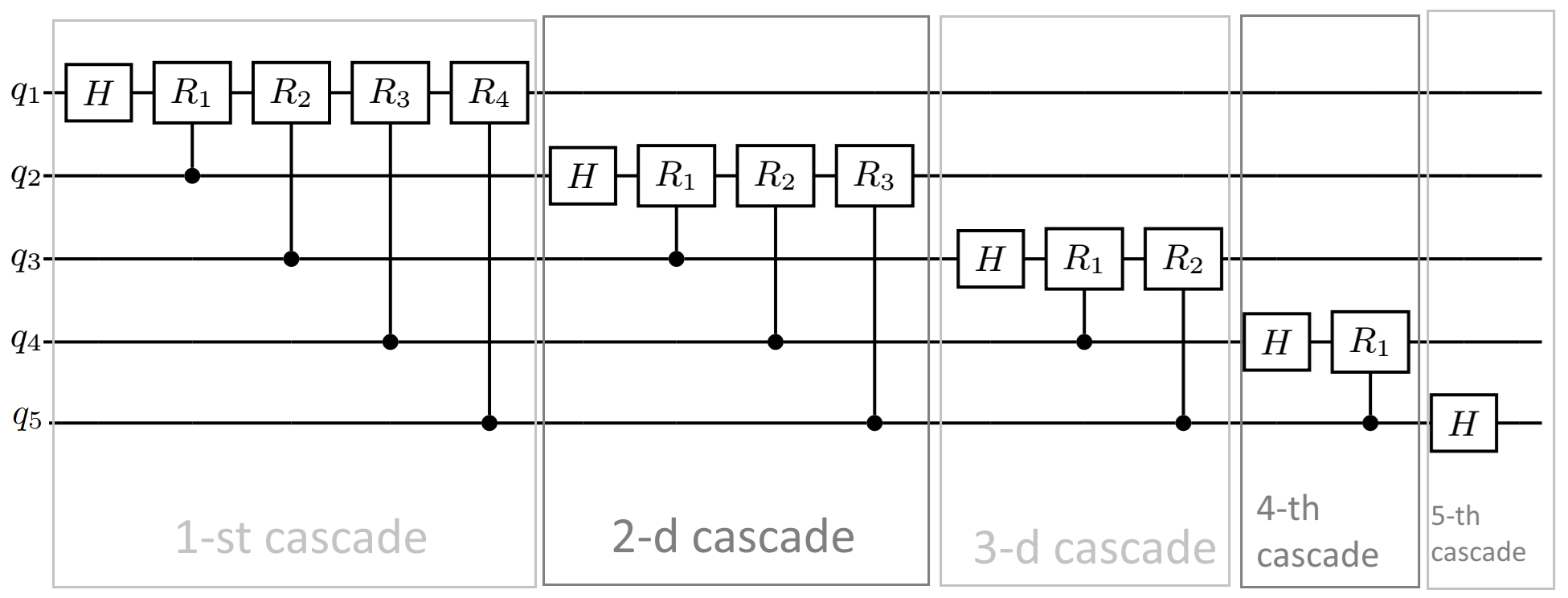}
\caption{\label{fig:qft-cascade} A quantum circuit for Quantum Fourier Transform algorithm for fully connected $5$ qubits splited to $5$ cascades depending on the target qubit.}
\end{figure}

Assume that we have a $\textsc{CascadeForPath}(P,r)$ procedure that constructs the $r$-th cascade of the circuit for the QFT algorithm for a path $P$. 
Here $P$ is a path that ``covers'' only vertices corresponding to the qubits used in the current cascade. We say that a path covers a vertex if the vertex is visited by the path or the vertex is connected by an edge with some vertex from the path. Because we can apply two-qubit gates only for adjacent vertices, the procedure moves the target qubit by the path $P$ from the first vertex of the $P$ to the last one. We move the target qubit using the SWAP gate. During the ``travel'' of the target qubit, we apply the control phase operator to each neighbor vertex. Because the path $P$ covers all the vertices that correspond to the cascade. This strategy allows us to implement the cascade. In the end of the ``travel'', we move the target qubit to one of the neighbors of the last vertex of $P$ and exclude it from the next steps because it does not participate in rest cascades.  

Firstly, we present the main algorithm in Section \ref{sec:qft1}. Then we present the detailed algorithm for the $\textsc{CascadeForPath}(P,r)$ procedure in Section \ref{sec:qft2}. After that we discuss the complexity of the circuit in Section \ref{sec:qft3}. Finally, we compare the circuit with existing results in Section \ref{sec:qft4}. 

\subsection{The Main Algorithm}\label{sec:qft1}
Let us present the entire algorithm for constructing the quantum circuit for the QFT algorithm. 

\subsubsection{Vertices and Qubits Correspondence}
Firstly, we should assign logical qubits to the vertices. Consider two sequences:
\begin{itemize}
    \item $A_1,\dots, A_n$ are the indexes of initial positions of qubits. If $A_i=j$ on some step, it means that the vertex $v_i$ contains a logical qubit that was in $v_j$ before starting the algorithm.
    \item $S_1,\dots S_n$ are the final positions of the qubits. If $S_i=j$, then the $j$-th logical qubit is located in the vertex $v_i$ before starting the algorithm.
\end{itemize}

Our main goal is to compute the sequence $S_1,\dots S_n$. Let us present the algorithm.

\begin{itemize}
    \item[] \textbf{Step 0.} We assign $A_i\gets i$ for each $i\in\{1,\dots,n\}$. Let $r\gets 1$ be the number of a cascade.
    \item[] \textbf{Step 1.} We find a (3,2,1)-covering path $P_r=(v_{i_1},\dots,v_{i_k})$.
    \item[] \textbf{Step 2.} We assign $S_{A_{i_1}}\gets r$ 
    \item[] \textbf{Step 3.} We move the first element by the path, i.e. we swap $A_{i_j}$ and $ A_{i_{j+1}}$ for $j\in\{1,\dots,k-1\}$.
    \item[] \textbf{Step 4.} We choose a neighbor vertex $v_q$ of $v_{i_k}$ with the maximal index that is not visited by the path $P$. Then we assign $A_{i_k}\gets A_{q}$, and we exclude the vertex $v_q$ from the graph\footnote{In fact, we do not exclude it, but mark as excluded. After invocation of this algorithm, we should be able to restore the whole graph.}.
    \item[] \textbf{Step 5.} We go to the next cascade $r\gets r+1$. If $r\leq n-2$, then we go to Step 1, and go to Step 6 otherwise.
    \item[] \textbf{Step 6.} In this step, we have two vertices in the graph that are not excluded and connected. Assume that there are $v_q$ and $v_t$, and $q<t$. Then, we assign $S_{A_q}\gets n-1$, and $S_{A_t}\gets n$.  
\end{itemize}

The implementation of the algorithm is presented in Algorithm \ref{alg:compute-s}. (See Appendix \ref{apx:compute-s}).

\subsubsection{The Algorithm}
The enumeration $S$ is such that the algorithm works well, and the algorithm for computing $S$ is very similar to the main algorithm.

First, we restore the graph. Then, on each cascade, the $r$-th logical qubit is located at the starting vertex of the path $P_r$. For each cascade, we move the $r$-th logical qubit by the path $P_r$ using the SWAP gate and then to the neighbor of the last vertex of the path with the maximal index. After that, we exclude the qubit from the graph. 

We use $Q_i$ as the current position of the $i$-th logical qubit and $T_j$ as an index of logical qubit located in the vertex $v_j$. Initially $T_j\gets S_j$, $Q_{T_j}\gets j$ for each $j\in\{1,\dots, n\}$.

The construction of a cascade is presented by the procedure $\textsc{CascadeForPath}(P_r,r)$. The algorithm is as follows.

\begin{itemize}
    \item[] \textbf{Step 0.} We associate the $S_j$-th logical qubit with the vertex $v_j$, i.e. $T_j\gets S_j$, $Q_{T_j}\gets j$, for $j\in\{1,\dots,n\}$. 
    
    Let $r\gets 1$ be the number of a cascade.
    
    \item[] \textbf{Step 1.} We construct the $r$-the cascade using $\textsc{CascadeForPath}(P_r,r)$ and keep the $T$ and $Q$ indexes actual.
    \item[] \textbf{Step 2.} We choose a neighbor vertex $v_q$ of $v_{i_k}$ with the maximal index that is not visited by the path $P$ and exclude it because the $r$-th qubit was moved there during the $\textsc{CascadeForPath}(P_r,r)$ procedure.
    \item[] \textbf{Step 3.} We go to the next cascade $r\gets r+1$. If $r\leq n$, then we go to Step 1, and stop otherwise.
\end{itemize}

The implementation of the algorithm is presented in Algorithm \ref{alg:qft-whole}.  Assume that the $\textsc{ConstructS}(G)$ procedure contains Algorithm \ref{alg:compute-s}.
\begin{algorithm}[H]
\caption{Implementation of the algorithm of computing the sequence of indexes $S_1,\dots,S_n$.}\label{alg:compute-s}
\begin{algorithmic}
\For{$j\in\{1,\dots,n\}$}
\State $A_j\gets j$
\EndFor
\For{$r\in\{1,\dots,n-2\}$}
\State $(i_1,\dots,i_k)=P_r\gets \textsc{ThreeTwoOneCP}(G)$
\State $S_{A_{i_1}}\gets r$
\For{$j\in\{1,\dots,k-1\}$}
\State $x\gets A_{i_j}$, $ A_{i_j}\gets A_{i_{j+1}}$, $A_{i_{j+1}}\gets x$
\EndFor
\State $q=\max\{j:$ $v_j$ is not excluded,$ v_j\in \textsc{Neighbors}(v_{i_k}), j\neq i_{k-1}\}$
\State $A_{i_k}\gets A_q$
\State exclude $v_q$ from the graph.
\EndFor
\State $v_q$ and $v_t$ are two not excluded vertexes, and $q<t$
\State $S_{A_q}\gets n-1$, $S_{A_t}\gets n$
\State $P_{n-1}=(q)$, $P_{n}=()$
\end{algorithmic}
\end{algorithm}

\begin{algorithm}[H]
\caption{Implementation of the algorithm of constructing the whole circuit for QFT}\label{alg:qft-whole}
\begin{algorithmic}
\State $\textsc{ConstructS}(G)$
\For{$j\in\{1,\dots,n\}$}
    \State $T_j\gets S_j$
    \State $Q_{T_j}\gets j$
\EndFor
\For{$r\in\{1,\dots,n\}$}
    \State $\textsc{CascadeForPath}(P_r,r)$
    \State $P_r=(i_1,\dots,i_k)$
  \State $q=\max\{j:$ $v_j$ is not excluded,$ v_j\in \textsc{Neighbors}(v_{i_k}), j\neq i_{k-1}\}$
    \State exclude $v_q$ from the graph.
\EndFor
\end{algorithmic}
\end{algorithm}

Let us discuss the time complexity of the algorithm.

\begin{theorem}\label{th:qft-whole}
    The time complexity of Algorithm \ref{alg:qft-whole} is $O((m+n)2^n)$ in the case of exact solution and $O(mn\log n+n^2\log n)$ in the case of approximate solution.
\end{theorem}
\begin{proof}
The procedure \textsc{ConstructS}() invokes the algorithm for searching the (3,2,1)-covering path in the graphs of sizes $n,n-1,\dots, 1$.
In the case of an exact solution, the complexity of the procedure is at most \[O((m+n)2^n+(m+n-1)2^{n-1}+\dots+(m+n-n+1)2^{n-n+1})=O((m+n)\sum_{r=1}^{n}2^{r})=O((m+n)2^n).\]

In the case of an approximate solution, the complexity of the procedure is at most \[O((m+n)\log n+(m+n-1)\log (n-1)+\dots+(m+n-n+1)=O((m+n)\log n \cdot\sum_{r=1}^{n}r)=O((m+n)n\log n)=O(mn\log n + n^2\log n).\] 

The complexity of the rest part is at most $O(n^2)$.
So, the total complexity is $O((m+n)2^n+n^2)=O((m+n)2^n)$ in the case of exact solution; and $O(mn\log n + n^2\log n)$ in the case of the approximate solution.
\end{proof}

\subsection{Quantum Circuit for One Cascade}\label{sec:qft2}

Let us present the algorithm for generating a quantum circuit for the $r$-th cascade, that is the procedure $\textsc{CascadeForPath}(P,r)$.

In the $r$-th cascade, we use the $r$-th qubit as a target for the control phase gates. Due to the enumeration of qubits, it is located in the vertex $v_{i_1}$, where $P=(i_1,\dots,i_k)$. 

We move the target qubit by the path $P$ and for each position of the target qubit, we apply control phase gates for each neighbor vertex. Finally, we move the target qubit to the neighbor of $v_{i_k}$ with the maximal index. For refusing repetition of applying of a control phase gate for a control qubit, we use a set $U$ that stores all qubits that have already been used as control qubits during this cascade.

The algorithm for constructing a quantum circuit is as follows.

\begin{itemize}
    \item[] \textbf{Step 1.} We start with the first qubit in the path $j\gets 1$, and initialize $U\gets\emptyset$. We apply the Hadamard transformation to the qubit corresponding to the vertex $v_{i_1}$. We denote this action by $\textsc{H}(v_{i_1})$. If $k=1$, then we terminate our algorithm; otherwise, go to Step 2. 
      \item[]  \textbf{Step 2.} For each $v_t\in \textsc{Neighbors}(v_{i_j})\backslash\{v_{i_{j+1}}\}$, if $v_t\not\in U$, then we apply the control phase gate $CR_d$ with the control $v_{t}$ and the target $v_{i_j}$ qubits, where $d=T_{t}-r$. Note that $v_t$ with the maximal index should be processed in the end.      
      Then, we add $v_{t}$ to the set $U$, i.e. $U\gets U\cup\{v_{t}\}$. If $j=k$, then we go to Step 5, and to Step 3 otherwise.
        \item[]  \textbf{Step 3.} If $v_{i_{j+1}}\not\in U$, then we apply the control phase gate $CR_d$ with the control $v_{i_{j+1}}$ and the target $v_{i_j}$ qubits, where $d=T_{i_{j+1}}-r$. Then, we add $v_{i_{j+1}}$ to the set $U$, i.e. $U\gets U\cup\{v_{i_{j+1}}\}$. After that, we go to Step 4.
     \item[]  \textbf{Step 4.} We apply the SWAP gate to $v_{i_{j}}$ and $v_{i_{j+1}}$, and swap the indexes of qubits for these vertices. In other words, if $w_1=T_{i_j}$ and $w_2=T_{i_{j+1}}$ are indexes of the corresponding logical qubits, then we swap $Q_{w_1}$ and $Q_{w_2}$ values, and $T_{i_j}$ and $T_{i_{j+1}}$ values. Then, we update $j\gets j+1$ because the value of the target qubit moves to $v_{i_{j+1}}$. Then, we go to Step 2.
     \item[] \textbf{Step 5.} If $j=k$, then we apply the SWAP gate to $v_{i_{j}}$ and $v_{q}$, and swap the qubit indexes for these vertices similarly to Step 4. Here $v_q$ is the neighbor of $v_{i_{j}}$ with the maximal index, i.e. $q=\max\{j:$ $v_j$ is not excluded,$ v_j\in \textsc{Neighbors}(v_{i_k}), j\neq i_{k-1}\}$
\end{itemize}

Finally, we obtain the $\textsc{CascadeForPath}(P,r)$ procedure whose implementation is presented in Algorithm \ref{alg:qft-path} (see Appendix \ref{apx:qft-path}). This procedure constructs the $r$-th part (cascade) of the circuit for QFT for the path $P$.

\subsection{The CNOT cost of the Circuit}\label{sec:qft3}
Note that the $CR_d$ gate can be represented using only two CNOT gates and three $R_z$ gates \cite{barenco1995elementary} (see Figure \ref{fig:cp}).

\begin{figure}[H]
\includegraphics[width=0.6\textwidth]{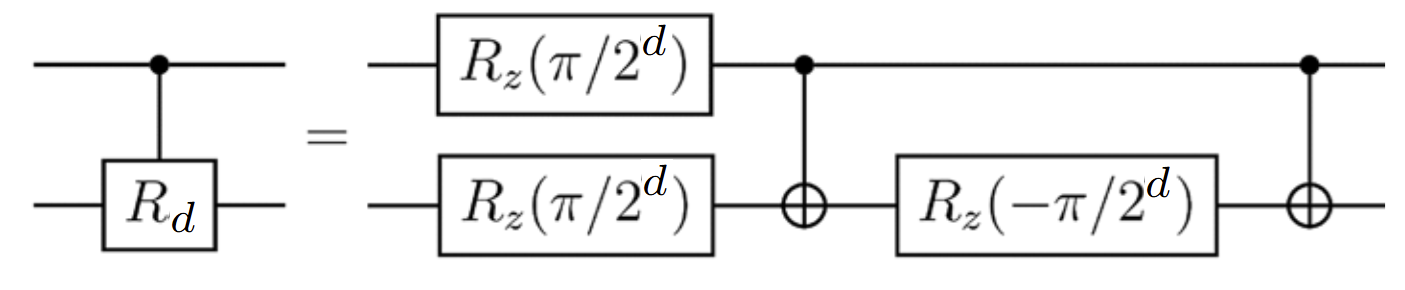}
\caption{\label{fig:cp} Representation of $CR_d$ gate using only basic gates}
\end{figure}  
A pair of $CR_d$ and $SWAP$ gates can be represented using three CNOT gates (see Figure \ref{fig:crswap}).

\begin{figure}[H]
\includegraphics[width=0.6\textwidth]{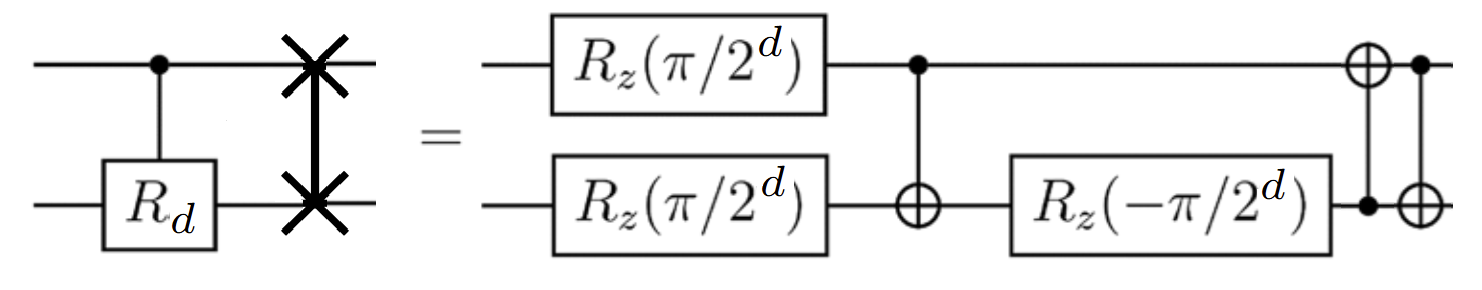}
\caption{\label{fig:crswap} Reduced representation of a pair $CR_d$ and $SWAP$ gates using only basic gates}
\end{figure}  

Let us discuss the CNOT cost of the algorithm in the next theorem.
\begin{theorem}\label{th:qft}
    The CNOT cost of the circuit that is generated using Algorithm \ref{alg:qft-whole} is at most $K+n^2-n-1$, where $K=\sum_{r=1}^{n-1}len(P_r)$ is the sum of lengths of the (3,2,1)-covering paths $P_r$.
\end{theorem}
\begin{proof}
Let us show that the CNOT cost of $r$-th cascade is at most $len(P_r)+2(n-r)$. We apply $CR_d$ and SWAP gates for each element of the path $P_r$ and the neighbor of $v_{i_k}$ with the maximal index. If we visit a vertex more than once, then we apply only the SWAP gate. Both operations have a CNOT cost $3$. So, their complexity is $3len(P_r)$. For all other vertices, we apply only the $CR_d$ gate whose CNOT cost is $2$. In the $r$-th cascade, we have already excluded $r-1$ vertices. So, there are $n-r-len(P_r)$ rest vertices. The total CNOT cost of the $r$-th cascade is

\[3len(P_r)+2(n-r-len(P_r))=len(P_r)+2(n-r)\]

The cascade $n-1$ has the CNOT cost $2$ that can be represented as $len(P_r)+2(n-r)-1$ for $r=n-1$.
The cascade $n$ has the CNOT cost $0$. The total CNOT cost is 
\[\sum_{r=1}^{n-1}(len(P_r)+2(n-r))-1 =\sum_{r=1}^{n-1}len(P_r) +\sum_{r=1}^{n-1}(2(n-r))-1=K+n^2-n-1. \]
\end{proof}



We have two corollaries from this result. Firstly, we can estimate $K$ as $nk - 0.5k^2+1.5k$, where $k$ is the length of a (3,2,1)-covering path in the graph $G$. We present this result in Corollary $\ref{cr:maxK}$. Then, we obtain the minimal and maximal bounds for the CNOT cost in Corollary \ref{cr:bounds}. 

\begin{corollary}\label{cr:maxK}
 The CNOT cost of the circuit that is generated using Algorithm \ref{alg:qft-whole} is at most $nk-0.5k^2-1.5k+n^2-n$, where $k$ is the length of a (3,2,1)-covering path in the graph $G$. 
\end{corollary}

\begin{proof}
In the worst case, the first $n-k-2$ cascades do not decrease the size of the (3,2,1)-covering paths, and $len(P_1)=\dots=len(P_{n-k-1})=k$. After that, we obtain a chain in which we have only vertices of the path $P_{n-k-1}=(v_{i_1},\dots,v_{i_k})$ and two vertices: one of them connected with $v_{i_1}$, and the second one is connected with $v_{i_k}$.



Then, the length of the paths decreases by $1$ for each next cascade, and $len(P_{r})=n-r-1$ for $n-k\leq r\leq n-2$, $len(P_{n-1})=1$. The final sum is
\[K=(n-k-1)k + 1 +\sum_{r=n-k}^{n-2}(n-r-1)=nk-k^2-k +1+ 0.5k^2-0.5k=
nk-0.5k^2-1.5k+1.\]

Due to Theorem \ref{th:qft}, the complexity is at most $nk-0.5k^2-1.5k+1+n^2-n-1 = nk-0.5k^2-1.5k+n^2-n$.

\end{proof}

\begin{corollary}\label{cr:bounds}
The CNOT cost of a circuit that is generated using Algorithm \ref{alg:qft-whole} is in the range between $n^2-2n-2$ and $2n^2-2n-2$. 
\end{corollary}
\begin{proof}
We can say that the length of the $P_r$ path is at most twice the number of vertices except two (in the beginning and at the end of the path), that is, $2n-2r$ due to Lemma \ref{lm:len-wnsh}, for $1\leq r\leq n-2$. At the same time, the minimal value is $1$ because the graph can be like a star (all vertices are connected to one), and the path is always the center of the star. The length $len(P_{n-1})=1$, and $len(P_{n})=0$ always. 

So, if $1\leq len(P_r) \leq 2n-2r$, then $n-1\leq K\leq \sum_{r=1}^{n-2}(2n-2r)+1=n^2-n-1$.

Due to Theorem \ref{th:qft}, CNOT cost of the circuit is in the range $n-1+n^2-n-1=n^2-2n-2$ and $n^2-n-1+n^2-n-1=2n^2-2n-2$.
\end{proof}

Let us make several remarks.
\begin{enumerate}
    \item If we use the approximate solution to the (3,2,1)-covering path problem, then the length of the (3,2,1)-covering path can be longer, but it cannot be longer than $2n$.
    \item When we say ``approximate'' solution, we do not mean approximate circuit for the QFT algorithm, but we mean approximate algorithm for constricting (3,2,1)-covering path that can give a larger quantum circuit with larger CNOT cost.
    \item The maximal number of neighbors $\Delta$ in current devices is often small (it can be $2,3,4$ or $5$ if we consider IBM or Regetti quantum devices). That is why $ln \delta$ can be a very small number.
    \item The cost of a (3,2,1)-covering path and the CNOT cost of a corresponding circuit for a cascade differ only in $1$. That is why the minimization of cost leads us to the minimization of CNOT cost of the circuit.  
\end{enumerate}

\subsection{Comparing With Other Results}\label{sec:qft4}
The most popular type of qubit connectivity graphs is the LNN architecture.
In that case, the graph is a chain, where a vertex $v_i$ is connected to $v_{i-1}$ and $v_{i+1}$.
For the architecture, the path visits all vertices from $v_2$ to $v_{n-1}$ one by one.
The circuit produced by our method is similar to the circuit developed in \cite{k2024aliya}. The length of the $P_r$ path is $n-r-1$, and $len(P_{n-1})=1$. Due to Theorem \ref{th:qft}, we get the following CNOT cost for the LNN architecture.
\begin{corollary}
The CNOT cost of the produced circuit for the QFT algorithm using $n$ qubits for the LNN architecture is $1.5n^2-2.5n+1$.
\end{corollary}

It is the same CNOT cost as for the circuit from \cite{k2024aliya}.
The CNOT cost for the circuit from \cite{kkcw2025} is $1.5n^2-1.5n-1$.
At the same time, \cite{park2023reducing} gives the circuit with the CNOT cost $n^2+n-4$. Our circuit (like the circuit from \cite{k2024aliya}) is better than \cite{park2023reducing} only if $n\leq 5$. However, it is a reasonable restriction for current and near-future devices. If we look at one of the QFT applications which is the quantum phase estimation (QPE) algorithm \cite{k1995}, then we can see that $n$ is the precision of the phase estimation. In that case, $5$ bits is already a reasonable value.  
However, it is not known how to apply the results of \cite{park2023reducing} to more complex architecture. Note that our result is always better than the circuit from \cite{kkcw2025}. 

Secondly, let us consider more complex architectures like 16-qubit ``sun''  (Figure \ref{fig:sun1}, the left one), and 27-qubit  ``two joint suns'' (Figure \ref{fig:sun1}, the right one). 
The results circuit is the same as in \cite{k2024aliya}. The CNOT cost for the 16-qubit machine is $324$, and for the 27-qubit machine is $957$.
\begin{figure}[H]
\includegraphics[width=0.4\textwidth]{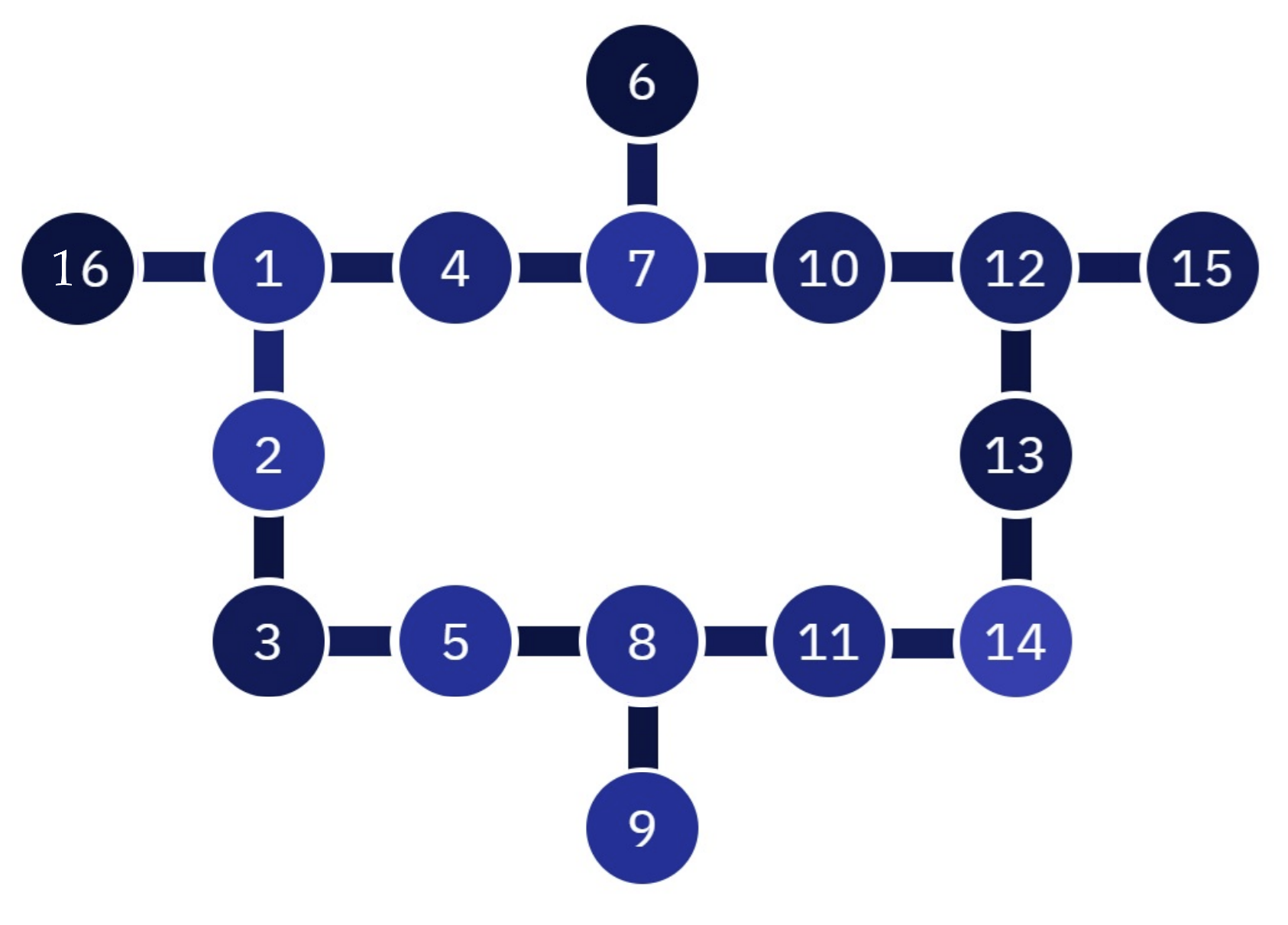}$\quad$
\includegraphics[width=0.6\textwidth]{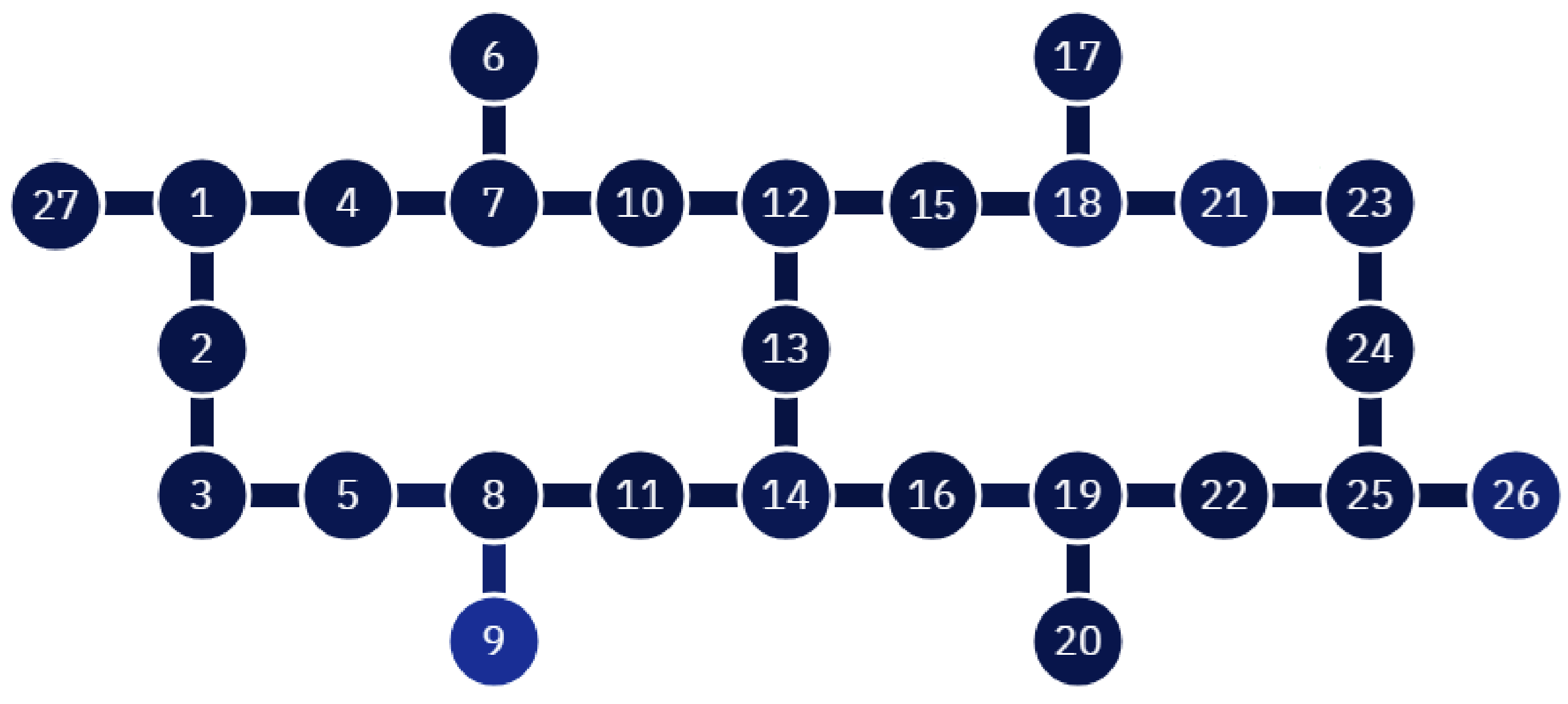}
\caption{\label{fig:sun1} ``Sun'' (16-qubit IBMQ Falcon r4P) architecture on the left. ``Two joint suns'' (27-qubit IBMQ Falcon r5.11) architecture on the right.}
\end{figure}

So, our generic method gives better circuits than the circuits generated by \cite{kkcw2025}, which CNOT costs are $342$ and $1009$ for 16-qubit and 27-qubit architectures, respectively. The difference between results is about 5\%.
\section{Conclusion}\label{sec:concl}
We present a generic method for constructing quantum circuits for the quantum Fourier transform algorithm for implementation on hardware with an arbitrary architecture of qubit connection. The method has $O((m+n)2^n)$ time complexity (and $O(mn\log n)$ in the case of the approximate solution) and it works for arbitrary connected graphs. Note that when we say ``approximate'' solution, we do not mean an approximate circuit for the QFT algorithm, but we mean an approximate algorithm for constricting (3,2,1)-covering path that can give us a quantum circuit with a larger CNOT cost.

Moreover, if we consider samples of graphs like ``sun'' (16-qubit IBMQ Falcon r4P architecture), and ``two joint suns'' (27-qubit IBMQ Falcon r5.11 architecture), then our generic algorithm gives us the same circuit as optimized especially for these graphs \cite{k2024aliya}. In the case of the LNN architecture, our algorithm gives a bit worse circuit compared to the technique optimized for these graphs \cite{park2023reducing}. At the same time, our approach works for arbitrary connected graphs, but the existing results work only for some specific graphs. 

Furthermore, our technique gives better results than the existing technique for arbitrary graphs \cite{kkcw2025}.

An open question is to develop a technique for QFT for an arbitrary connected graph that gives us the same or better results than the existing ones for LNN. The presented work gives a positive answer to similar questions for ``sun'' (16-qubit IBMQ Falcon r4P architecture), and ``two joint suns'' (27-qubit IBMQ Falcon r5.11 architecture) that were suggested in \cite{kkcw2025}.

%
%
%
 \bibliographystyle{splncs04}
 \bibliography{tcs}

\begin{thebibliography}{10}
\providecommand{\url}[1]{\texttt{#1}}
\providecommand{\urlprefix}{URL }
\providecommand{\doi}[1]{https://doi.org/#1}

\bibitem{aazksw2019part1}
Ablayev, F., Ablayev, M., Huang, J.Z., Khadiev, K., Salikhova, N., Wu, D.: On
  quantum methods for machine learning problems part i: Quantum tools. Big Data
  Mining and Analytics  \textbf{3}(1),  41--55 (2019)

\bibitem{a2017}
Ambainis, A.: Understanding quantum algorithms via query complexity. In: Proc.
  Int. Conf. of Math. 2018. vol.~4, pp. 3283--3304 (2018)

\bibitem{barenco1995elementary}
Barenco, A., Bennett, C.H., Cleve, R., DiVincenzo, D.P., Margolus, N., Shor,
  P., Sleator, T., Smolin, J.A., Weinfurter, H.: Elementary gates for quantum
  computation. Physical review A  \textbf{52}(5), ~3457 (1995)

\bibitem{best1996}
Barenco, A., Ekert, A., Suominen, K.A., T{\"o}rm{\"a}, P.: Approximate quantum
  fourier transform and decoherence. Physical Review A  \textbf{54}(1), ~139
  (1996)

\bibitem{vs2020}
van Bevern, R., Slugina, V.A.: A historical note on the 3/2-approximation
  algorithm for the metric traveling salesman problem. Historia Mathematica
  \textbf{53},  118--127 (2020)

\bibitem{bbwdr2019}
Bhattacharjee, A., Bandyopadhyay, C., Wille, R., Drechsler, R., Rahaman, H.:
  Improved look-ahead approaches for nearest neighbor synthesis of 1d quantum
  circuits. In: 2019 32nd International Conference on VLSI Design and 2019 18th
  International Conference on Embedded Systems (VLSID). pp. 203--208. IEEE
  (2019)

\bibitem{bhmt2002}
Brassard, G., H{\o}yer, P., Mosca, M., Tapp, A.: Quantum amplitude
  amplification and estimation. Contemporary Mathematics  \textbf{305},  53--74
  (2002)

\bibitem{c2022}
Christofides, N.: Worst-case analysis of a new heuristic for the travelling
  salesman problem. In: Operations Research Forum. vol.~3, p.~20. Springer
  (2022)

\bibitem{cormen2001}
Cormen, T.H., Leiserson, C.E., Rivest, R.L., Stein, C.: Introduction to
  Algorithms. McGraw-Hill (2001)

\bibitem{cpr1994}
Current, J., Pirkul, H., Rolland, E.: Efficient algorithms for solving the
  shortest covering path problem. Transportation Science  \textbf{28}(4),
  317--327 (1994)

\bibitem{dg1997}
Dorigo, M., Gambardella, L.M.: Ant colonies for the travelling salesman
  problem. Biosystems  \textbf{43}(2),  73--81 (1997).
  \doi{https://doi.org/10.1016/S0303-2647(97)01708-5}

\bibitem{d2000}
Draper, T.G.: Addition on a quantum computer. arXiv preprint quant-ph/0008033
  (2000)

\bibitem{fdh2004}
Fowler, A., Devitt, S., Hollenberg, L.: Implementation of shor's algorithm on a
  linear nearest neighbour qubit array. Quantum Information \& Computation
  \textbf{4}(4),  237--251 (2004)

\bibitem{gk1998}
Guha, S., Khuller, S.: Approximation algorithms for connected dominating sets.
  Algorithmica  \textbf{20},  374--387 (1998)

\bibitem{hhl2009}
Harrow, A.W., Hassidim, A., Lloyd, S.: Quantum algorithm for linear systems of
  equations. Physical review letters  \textbf{103}(15),  150502 (2009)

\bibitem{jm1997}
Johnson, D.S., McGeoch, L.A.: The traveling salesman problem: a case study.
  Local search in combinatorial optimization pp. 215--310 (1997)

\bibitem{quantumzoo}
Jordan, S.: Quantum algorithms zoo (2023), http://quantumalgorithmzoo.org/

\bibitem{k2022lecturenotes}
Khadiev, K.: Lecture notes on quantum algorithms. arXiv preprint
  arXiv:2212.14205  (2022)

\bibitem{kkcw2025}
Khadiev, K., Khadieva, A., Chen, Z., Wu, J.: Implementation of quantum fourier
  transform and quantum hashing for a quantum device with arbitrary qubits
  connection graphs. arXiv preprint arXiv:2501.18677  (2025)

\bibitem{k2024aliya}
Khadieva, A.: Quantum hashing algorithm implementation. arXiv preprint  (2024),
  arXiv:quant-ph/2024

\bibitem{k1995}
Kitaev, A.Y.: Quantum measurements and the abelian stabilizer problem. arXiv
  preprint quant-ph/9511026  (1995)

\bibitem{kds2017}
Kole, A., Datta, K., Sengupta, I.: A new heuristic for $ n $-dimensional
  nearest neighbor realization of a quantum circuit. IEEE Transactions on
  Computer-Aided Design of Integrated Circuits and Systems  \textbf{37}(1),
  182--192 (2017)

\bibitem{nc2010}
Nielsen, M.A., Chuang, I.L.: Quantum computation and quantum information.
  Cambridge univ. press (2010)

\bibitem{pa2022}
Park, B., Ahn, D.: T-count optimization of approximate quantum fourier
  transform. arXiv preprint arXiv:2203.07739  (2022)

\bibitem{park2023reducing}
Park, B., Ahn, D.: Reducing cnot count in quantum fourier transform for the
  linear nearest-neighbor architecture. Scientific Reports  \textbf{13}(1),
  ~8638 (2023)

\bibitem{swd2011}
Saeedi, M., Wille, R., Drechsler, R.: Synthesis of quantum circuits for linear
  nearest neighbor architectures. Quantum Information Processing  \textbf{10},
  355--377 (2011)

\bibitem{s1978}
Serdyukov, A.: On some extremal walks in graphs (in russian). Upravlyaemye
  sistemy (17),  76--79 (1978)

\bibitem{s1999}
Shor, P.W.: Polynomial-time algorithms for prime factorization and discrete
  logarithms on a quantum computer. SIAM review  \textbf{41}(2),  303--332
  (1999)

\bibitem{tko2007}
Takahashi, Y., Kunihiro, N., Ohta, K.: The quantum fourier transform on a
  linear nearest neighbor architecture. Quantum Information \& Computation
  \textbf{7}(4),  383--391 (2007)

\bibitem{wld2014}
Wille, R., Lye, A., Drechsler, R.: Exact reordering of circuit lines for
  nearest neighbor quantum architectures. IEEE Transactions on Computer-Aided
  Design of Integrated Circuits and Systems  \textbf{33}(12),  1818--1831
  (2014)

\end{thebibliography}

 \newpage
\appendix

\appendix

\section{Implementation of $\textsc{GetNSPath}(S,v)$}\label{apx:getpath2}

\begin{algorithm}[H]
\caption{Implementation of $\textsc{GetNSPath}(S,v)$}\label{alg:getpath2}
\begin{algorithmic}
\State $P=()$ \Comment{We initialize it by an empty list}
\While{$F(S,v)\neq NULL$}
\State $u\gets F(S,v)$, $S\gets S\backslash\{v\}$
\State $P\gets \textsc{GetShortestPath}(u,v)\circ P$\Comment{We add $P_{u,v}$ path without the vertex $u$ to the begin of the list}
\State $v\gets u$
\EndWhile
\State $P\gets v\circ P$
\State \Return $P$
\end{algorithmic}
\end{algorithm}



\section{Implementation of $\textsc{GetShortestPath}(v,u)$}
\label{apx:getpath}
\begin{algorithm}[H]
\caption{Implementation of $\textsc{GetShortestPath}(v,u)$}\label{alg:getpath}
\begin{algorithmic}
\State $t\gets A[v,u]$
\State $P_{v,u}\gets (u)$
\While{$t\neq v$}
\State $P_{v,u}\gets t\circ P_{v,u} $
\State $t\gets  A[v,t]$
\EndWhile
\State \Return $P_{v,u}$
\end{algorithmic}
\end{algorithm}

\section{Implementation of the Procedure $\textsc{ShortestPathes}$ for Shortest Paths Searching}\label{apx:floyd}
Here we discuss how to construct matrices $W$ and $A$ such that $W[v,u]$ is the length of the shortest path between vertices $v$ and $u$, and $A[v,u]$ is the last vertex in the shortest path between $v$ and $u$. The procedures are simple, but we present them for the completeness of the results representation.

Firstly, we present a procedure $\textsc{SingleSrcShortestPath}(v)$ that finds the shortest paths for a single source vertex $v$ that is based on the BFS algorithm \cite{cormen2001}. The algorithm calculates the $v$-th rows of $W$ and $A$. The implementation is presented in Algorithm \ref{alg:bfs1}. Here we assume that we have a queue data structure \cite{cormen2001} that allows us to do the next actions in constant time:
\begin{itemize}
    \item $\textsc{Add}(queue, v)$ adds an element to the queue;
    \item $\textsc{Remove}(queue)$ removes an element from the queue and returns the element;
    \item  $\textsc{Init}()$ returns an empty queue;
     \item $\textsc{isEmpty}(queue)$ returns $True$ if the queue is empty and $False$ otherwise.
\end{itemize}
\begin{algorithm}[H]
\caption{Implementation of $\textsc{SingleSrcShortestPath}(v)$}\label{alg:bfs1}
\begin{algorithmic}
\State $queue\gets \textsc{Init}()$
\State  $\textsc{Add}(queue, v)$
\For{$u\in V$}
\State $W[v,u]\gets\infty$
\State $A[v,u]\gets NULL$
\EndFor
\State $W[v,v]\gets 0$
\While{$\textsc{isEmpty}(queue)=False$}
\State $t\gets\textsc{Remove}(queue)$
\For{$r\in \textsc{Neighbors}(t)$}
\If{$W[v,r]=\infty$}
    \State $A[v,r]\gets t$
    \State $W[v,r]=W[v,t]+1$
    \State $\textsc{Add}(queue, r)$
\EndIf
\EndFor
\EndWhile
\end{algorithmic}
\end{algorithm}
As an implementation of the $\textsc{ShortestPaths}$ procedure, we invoke $\textsc{SingleSrcShortestPath}(v)$ for each vertex $v\in V$.
\begin{algorithm}[H]
\caption{Implementation of $\textsc{ShortestPaths}(G)$ for a $G=(V,E)$ graph}\label{alg:bfs2}
\begin{algorithmic}
\State 
\For{$v\in V$}
    \State $\textsc{SingleSrcShortestPath}(v)$
\EndFor
\State \Return $(W,A)$
\end{algorithmic}
\end{algorithm}
\begin{lemma}
    The time complexity of the $\textsc{ShortestPathes}$ procedure is $O(n^3)$.
\end{lemma} 
\begin{proof}
Time complexity of BFS is $O(n+m)=O(n^2)$ due to \cite{cormen2001}.  Invocation of $n$ BFS algorithms for each $v\in V$ is $O(n^3)$. 
\end{proof}

\section{Quantum Fourier Transform}\label{apx:qft}
QFT is a quantum version of the discrete Fourier transform. The definitions of $n$-qubit QFT and its inverse are as follows:
\[QFT|j\rangle = \sum_{k=0}^{2^n-1}e^{\frac{2\pi i jk}{2^n}}|k\rangle,\]
\[QFT^{-1}|j\rangle = \sum_{k=0}^{2^n-1}e^{-\frac{-2\pi i jk}{2^n}}|k\rangle,\]
The $n$-qubit QFT circuit requires $0.5n^2 - 0.5n$ control phase ($CR_d$) gates and $n$ Hadamard ($H$) gates if we have no restriction on the application of two-qubit gates (See Figure \ref{fig:cqft}). The $CR_d$ gate is represented by basic gates that require two CNOT and three $R_z$ gates \cite{barenco1995elementary}. Therefore, $n^2 - n$ CNOT gates are required to construct an $n$-qubit QFT circuit. At the same time, if a quantum device has the LNN architecture, then for implementing the QFT, the number of CNOT gates is much larger than $n^2 - n$ \cite{fdh2004,swd2011,wld2014,kds2017,bbwdr2019,park2023reducing}. If we consider a general graph, then the situation is much worse than \cite{k2024aliya}.


\section{Implementation of $\textsc{CascadeForPath}(P,r)$ procedure}
\label{apx:qft-path}
\begin{algorithm}[H]
\caption{Implementation of $\textsc{CascadeForPath}(P,r)$ procedure. Algorithm of constructing the circuit for the $r$-th cascade for the path $P=(v_{i_1}, \dots,v_{i_{k}})$}\label{alg:qft-path}
\begin{algorithmic}
\State $j\gets 1$\Comment{Step 1}
\State $\textsc{H}(v_{i_j})$
\State $U\gets \emptyset$
\While{$j\leq k$}
\For{$t\in \textsc{Neighbors}(v_{i_j})\backslash\{v_{i_{j+1}}\}$}\Comment{Step 2}
\If{$v_t\not\in U$}
\State $d\gets T_t-r$
\State $\textsc{CR}_d(v_t,v_{i_j})$
\State $U\gets U\{v_t\}$
\EndIf
\EndFor
\If{$j\leq k-1$}
\If{$v_{i_{j+1}}\not\in U$}\Comment{Step 3}
\State $d\gets T_{i_{j+1}}-r$
\State $\textsc{CR}_d(v_{i_{j+1}},v_{i_j})$
\State $U\gets U\{v_{i_{j+1}}\}$
\EndIf

\State $\textsc{swap}(v_{i_j},v_{i_{j+1}})$\Comment{Step 4}
\State $w_1\gets T_{i_j}, w_2\gets T_{i_{j+1}}$
\State $Q_{w_1}\gets i_{j+1}$, $Q_{w_2}\gets i_{j}$ 
\State $T_{i_j}\gets w_2$, $T_{i_{j+1}}\gets w_1$
\Else
\State $q=\max\{j:$ $v_j$ is not excluded,$ v_j\in \textsc{Neighbors}(v_{i_k}), j\neq i_{k-1}\}$
\State $\textsc{swap}(v_{i_j},v_{q})$\Comment{Step 5}
\State $w_1\gets T_{i_j}, w_2\gets T_{q}$
\State $Q_{w_1}\gets q$, $Q_{w_2}\gets i_{j}$ 
\State $T_{i_j}\gets w_2$, $T_{q}\gets w_1$
\EndIf

\State $j\gets j+1$

\EndWhile
\end{algorithmic}
\end{algorithm}

\end{document}